\documentclass[runningheads]{llncs}

\usepackage{amsmath}
\usepackage[T1]{fontenc}
\usepackage{graphicx,psfrag,amssymb,enumerate,color,pdfpages,dsfont,float,comment,tabularx,subcaption,placeins,booktabs,caption}

\usepackage{todonotes}

\newcommand{\E}{{\mathbb E}}

\newcommand{\eqn}[1]{\begin{equation}#1\end{equation}}

\newcommand{\sss}{\scriptscriptstyle}

\begin{document}
\title{Degrees in preferential attachment networks\\ with an anomaly}

\author{Qiu Liang\inst{1}\orcidID{0009-0001-7530-7256} \and
Remco van der Hofstad\inst{1}\orcidID{0000-0003-1331-9697} \and
Nelly Litvak\inst{1}\orcidID{0000-0002-6750-3484}}
\authorrunning{Q. Liang et al.}
%
\institute{Department of Mathematics and Computer Science, Eindhoven University of Technology, Eindhoven, The Netherlands\\
\email{\{q.liang2, r.w.v.d.hofstad, n.v.litvak\}@tue.nl}}

\maketitle              

\begin{abstract}
We consider a preferential attachment model that incorporates an anomaly. Our goal is to understand the evolution of the network before and after the occurrence of the anomaly by studying the influence of the anomaly on the structural properties of the network. The anomaly is such that after its arrival it attracts newly added edges with fixed probability. We investigate the growth of degrees in the network, finding that the anomaly's degree increases almost linearly. We also provide a heuristic derivation for the exponent of the limiting degree distributions of ordinary vertices, and study the degree growth of the oldest vertex. We show that when the anomaly enters early, the degree distribution is altered  significantly, while a late anomaly has minimal impact. Our analysis provides deeper insights into the evolution of preferential attachment networks with an anomalous vertex.

\end{abstract}

\section{Introduction}

Dynamic network models, where nodes and edges appear or disappear over time, attracted a lot of attention in the network science literature. Among these models, the Preferential Attachment (PA) network, as presented by Barabási and Albert ~\cite{barabasi1999emergence}, has been particularly influential, because it explains the emergence of scale-free property in networks through the  `rich-get-richer' phenomenon. Specifically, the probability that a new vertex connects to an existing vertex is proportional to the degree of the existing vertex.

One extension of the standard PA network is the {\em superstar model} \cite{bhamidi2015superstar}, which is used to analyze key features of retweet networks. In this model, the initial vertex in the network is a {\em superstar}. At each time step, a newly added vertex connects to the {superstar} with probability $p$, or to one of the non-superstar vertices with probability $1-p$ according to the preferential attachment rule. The {superstar} alters the dynamics and the resulting degree distribution. It was shown in \cite{bhamidi2015superstar} that the non-superstar vertices follow a power-law degree distribution, though with a modified exponent $ 3 + \frac{p}{1-p}$. Additionally, the maximal degree of the non-superstar vertices is likewise affected; it grows slower with the network size. In this paper, we consider a PA model with an anomaly. Our anomaly is similar to the {\em superstar}, but it may arrive at any point of time.

Since the anomaly alters the network dynamics, our work is closely related to the line of research on PA models with a {\em change point}, where a parameter of the PA model changes at some point of time. For instance, \cite{banerjee2023fluctuation,bhamidi2018longrange} investigate methods to identify the change point in preferential attachment trees via embedding the discrete time tree in a continuous-time branching process and studying the proportion of leaves, while \cite{kay2023detecting} proposes an approach based on the fraction of vertices with minimal degree to detect a late change point. Furthermore, \cite{cirkovic2022likelihood} applies the likelihood ratio technique to estimate the change point in a PA model, and extend the method to detect multiple change points via screening and ranking, as well as  binary segmentation.

While {\em change point} detection addresses abrupt changes in the parameters of the network dynamics, here we instead focus on how structural properties evolve when a single anomalous vertex enters a PA network. To explore this, we suggest a new model incorporating an anomaly. Our model is an extension of the {\em superstar model}, and the two models are highly similar when the anomaly coincides with the {\em initial} vertex. Our main contributions are as follows:
\begin{enumerate}
    \item[$\rhd$] We propose a PA network with an anomaly. The network evolves according to the standard preferential attachment rule until the anomaly enters. Once the anomaly appears in the network, it attracts newly added edges with a fixed probability, plus a probability that depends on its current degree as in the normal PA dynamics.  
    \item[$\rhd$] We compute the mean degree of the anomaly as a function of network size, and study the mean degree of other  vertices and their convergence. 
    \item[$\rhd$] We provide a heuristic derivation of the limiting degree distribution of the ordinary vertices when the anomaly arrives in various different stages of the network's evolution. 
    \end{enumerate}

Our results serve as a first step towards the understanding of anomaly detection in PA networks. In the final section, we provide an outlook on this detection problem.

\section{Model description}
We start by introducing the preferential attachment network without an anomaly. The model constructs a graph sequence $\left( G_{t} \right)_{t \ge 2}$ such that each graph $ G_t $ is formed by adding one new vertex and $m$ edges connecting the new vertex to existing vertices, where $m \ge 1$. Let $ G_{t} = (V_{t}, E_{t})$, where $ V_{t} = \{v_1, \cdots, v_t \}$ and $ E_t \subseteq \{\{v,w\}\colon v,w \in V_t \}$.  The {\em initial} graph $G_1$ consists of a single vertex $v_1$ and $ m $ self-loops. For $t>1$, no self-loops are present. We treat the process of connecting each edge from a newly introduced vertex to an existing vertex in the network as an individual step. Specifically, let $ G_{t,j}$ denote the network after the $j$th edge of a newly added vertex $ v_{t}$ connects to an existing vertex $ v_i \in \left\{ v_1, v_2, \cdots, v_{t-1} \right\} $, where $ j \in [m]\equiv \{1, \ldots, m\}$. We let $ D_i (t,j)$ denote the degree of vertex $ i $ in $ G_{t,j}$, and introduce a {\em fitness} parameter $\delta > -m$. Further, we define $ G_{t} = G_{t,m} = G_{t+1,0}$, and $ D_i (t) = D_i (t,m)$. 
For $ t > 1$, $ j \in [m]$, we define the attachment rule for the $j$th edge linking to the vertex $ v_i \in \left\{ v_1, v_2, \cdots, v_{t-1} \right\} $ as defined in \cite{BerBorChaSab14,BerBorChaSab05}, and studied further in \cite{kay2023detecting}: 
\begin{equation}
    \label{rule-1}
    P(v_{t,j} \to v_i \mid  G_{t,j-1}) = \frac{D_i(t,j-1) + \delta}{2m(t-1)+(t-1)\delta + (j-1)}. 
\end{equation}
\smallskip

Assume that an anomaly occurs at some time $\tau$ satisfying $ 1< \tau < t$, where the attachment rule changes after the anomaly has occurred. We denote the anomaly by the vertex $v_\tau$. After $\tau$, each new edge connects to the anomaly with probability 
    \begin{equation}
    \label{p-beta-def}
    p \approx \frac{\beta}{ 2m+\beta+\delta},
    \end{equation}
where $\beta>0$ is a parameter of the model. Otherwise, with complementary probability, the edge connects to any existing vertex, including the anomaly, by following the usual PA rule. Formally, the dynamics of the model {\em at step $t$} are as follows:

\begin{enumerate}
    \item[(I)] If $ t<\tau $, the anomalous vertex has not occurred in the graph $G_t$, the attachment rule of $G_t$ is the same as \eqref{rule-1}.
    \item[(II)] If $ 1 < \tau \leq t$, the evolution rule changes as follows: 
\begin{equation}
\label{rule-2}
P(v_{t,j} \to v_i \mid G_{t,j-1})=
\begin{cases}
\displaystyle
 \frac{D_{i}(t,j-1)+\delta}{(t-1)(2 m+ \beta + \delta)+ j-1} & \text{ if } i \neq \tau ,\\
 \displaystyle\frac{(t-1) \beta+ D_{\tau}(t,j-1)+\delta}{(t-1)(2 m+ \beta + \delta)+ j-1
 } & \text{ if } i = \tau .
\end{cases}
\end{equation}
\end{enumerate}
To explain the rationale behind \eqref{rule-2}, suppose $m=1$, hence $j=1$ is the only edge of vertex $v_t$. Then
    \begin{align}
    \label{p-approx-precise}
    &\frac{(t-1) \beta+ D_\tau(t-1)+\delta}{(t-1)(2 + \beta + \delta)
    } = p+ (1-p)\frac{ D_\tau(t-1) +\delta}{(t-1)(2+\delta)},\nonumber
    \end{align}
with $p=\frac{\beta}{2+\beta+\delta}$. Hence,  we can think of our connection rule as connecting with probability $p$ to the anomaly, and proportional to the degrees (including the anomaly) with probability $1-p$. This explains the choice in \eqref{p-beta-def}, which is {\em exactly} correct for $m=1$. To make \eqref{p-beta-def} {\em exactly} correct for $m>1$, we have to choose $p=p_{\beta,t,j}$ slightly differently, so that 
    \begin{align*}
     &p_{\beta,t,j}+(1-p_{\beta,t,j}) \frac{D_\tau(t,j-1)+\delta}{(t-1)(2m+\delta)+j-1}\nonumber\\
     &\qquad=\frac{p_{\beta,t,j} \left[ (t-1)(2m+\delta)+(j-1) \right] + (1-p_{\beta,t,j})(D_\tau(t,j-1)+\delta)}{(t-1)(2m+\delta)+j-1}\nonumber\\
     &\qquad =\frac{(t-1) \beta+ D_{\tau}(t,j-1)+\delta}{(t-1)(2 m+ \beta + \delta)+ j-1},
    \end{align*}
which leads to $p_{\beta,t,j}\approx p$ in \eqref{rule-2}. Since the precise form of \eqref{rule-2} is a little simpler, we choose to work with this parameterization instead to simplify the formulas.

\smallskip

Figure \ref{fig:PA-example} illustrates an example of a PA network with an anomaly. We see that a large number of edges connect to the anomaly, yet a positive proportion of the edges of vertices arriving after $\tau$ are attached to ordinary vertices. Our goal is to study the structural properties of a PA network with an anomaly and the asymptotic degree distribution for different types of vertices.

\begin{figure}[t]
    \centering
    \begin{minipage}{0.45\textwidth}
        \includegraphics[width=\linewidth]{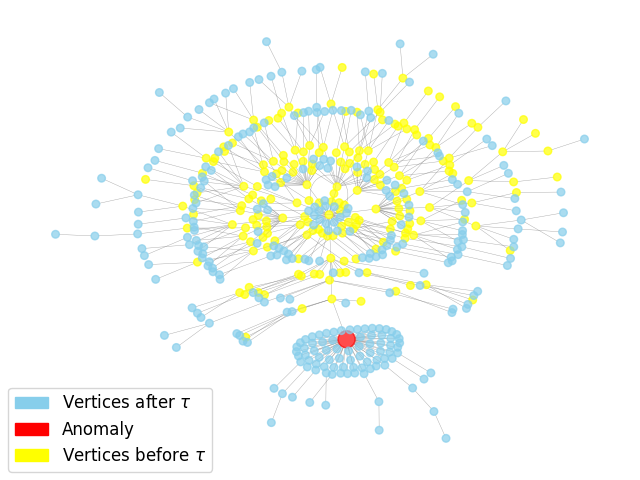}
        \subcaption{$\beta = 0.5$}
        \label{fig:PA-example - smaller beta}
    \end{minipage}
    \begin{minipage}{0.45\textwidth}
        \includegraphics[width=\linewidth]{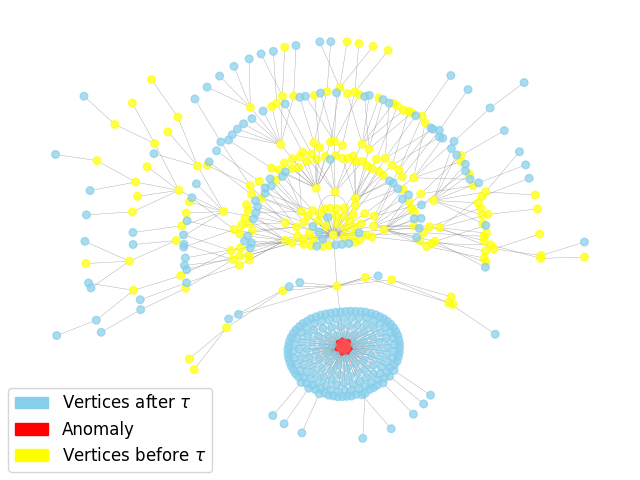}
        \subcaption{$\beta = 2.0$}
        \label{fig:PA-example - bigger beta}
    \end{minipage}
    \vspace{0.3cm}
    
    \caption{Examples of PA networks with an anomaly. Here $t =500, \tau = 200, \delta = 0, m = 1.$} 
    \label{fig:PA-example}
\end{figure}

\section{The growth of the degrees in $G_t$}

In this section, we analyze the growth of the degrees in $G_t$, both for the anomaly (Section~\ref{ssec:degree-anomaly}) and the ordinary vertices (Sections~\ref{ssec:degree-ordinary-expected}, \ref{ssec:degree-ordinary-convergence}). We follow the approach in \cite[Chapter 8]{Hofstad_2016}, adapted to our model with an anomaly.

\subsection{Expected degree of the anomaly}
\label{ssec:degree-anomaly}
We start by investigating the expected degree of the anomaly:
\begin{proposition}[Degree of the anomaly]
Consider an anomaly that occurs at time $\tau$, where $ 1< \tau < t$, and follows the attachment rule \eqref{rule-2} with given $ \delta > -m$, and $\beta > 0$. Then 
\begin{equation}
\begin{aligned}
\label{eq:D-tau}
    \E [D_{\tau} (t)+ \delta] = \frac{m \beta t}{m+\beta+\delta} + c_0 \frac{ \Gamma(t + \frac{m}{2m+\beta+\delta} ) \Gamma(\tau)}{\Gamma(t) \Gamma(\tau + \frac{m}{2m+\beta+\delta} ) }.
\end{aligned}
\end{equation}
where $ c_0 = m+\delta - \frac{m \beta \tau }{m+\beta+\delta}$. In particular, 
\begin{equation}
    \label{eq:D-tau-beta-infty}
\lim_{\beta \to \infty} \E [D_{\tau} (t)+ \delta ] = (t-\tau +1) m + \delta.\end{equation}
 
\end{proposition}
\begin{proof} Recall that $ \E [D_\tau(\tau) + \delta] = m + \delta$, for $ 1<\tau < t$. Based on \eqref{rule-2}, the expected degree of the anomaly satisfies the recursion
\begin{align*}
    \E & [D_{\tau} (t,m)+ \delta \mid D_{\tau} (t,m-1)] \\
    &= D_{\tau}(t,m-1) + \delta + \E [D_{\tau} (t,m) - D_{\tau} (t,m-1)\mid G_{t,m-1}] \\
    &= D_{\tau}(t,m-1) + \delta + \frac{(t-1) \beta+ D_{\tau}(t,m-1) +\delta}{(t-1)(2 m+ \beta + \delta)+ (m-1)} \\
    &= (D_{\tau}(t,m-1) + \delta) \left( 1 + \frac{1}{(t-1)(2 m+ \beta + \delta)+ (m-1)}  \right)\\
   & \quad + \frac{(t-1) \beta}{(t-1)(2 m+ \beta + \delta)+ (m-1)}.
\end{align*}
Taking expectation, and solving the recursion, gives that
\begin{align}
    \nonumber
    \E [D_{\tau} (t,m)+ \delta]
    &= (m+\delta) \prod_{k_2=\tau}^{t-1}\prod_{k_1 = 0}^{m-1} \left(1+\frac{1}{k_2(2m+\beta+\delta)+k_1}\right) + C^m(\beta,\delta,t) \\
    \nonumber
    & =(m+\delta) \prod_{k_2 = \tau}^{t-1}\frac{k_2 + \frac{m}{2m+\beta+\delta}}{k_2} + C^m(\beta,\delta,t) \\
    \label{eq:D-tau-Gamma+C}
    &= (m+\delta) \frac{\Gamma(t + \frac{m}{2m+\beta+\delta}) \Gamma(\tau)}{\Gamma(\tau + \frac{m}{2m+\beta+\delta} )\Gamma(t)} + C^m(\beta,\delta,t),
\end{align}
where $ C^m(\beta,\delta,t) $ is a function of $ \beta $, $ t $ and $ \delta$, that, with $ c = 2m+\beta+\delta$, equals
\begin{align*}
    C^m(\beta, \delta, t) & = \sum_{k=0}^{m-1} \left[ \frac{(t-1)\beta}{(t-1)c+k}\times \frac{(t-1)c+m}{(t-1)c+k+1} \right] \\
    & \quad + \frac{(t-1)c+m}{(t-1)c} \times \sum_{k=0}^{m-1} \left[ \frac{(t-2)\beta}{(t-2)c+k}\times \frac{(t-2)c+m}{(t-2)c+k+1}\right] +\cdots\\
    & \quad + \prod_{t_1 = \tau +1}^{t-1}\frac{t_1 c + m}{t_1 c}\times \sum_{k=0}^{m-1} \left[\frac{\tau \beta}{\tau c+k} \times \frac{\tau c+m}{\tau c+k+1} \right].
\end{align*}
For each time step, by the telescoping sum identity, 
\begin{align*}
    \sum_{k=0}^{m-1}\frac{t_1 \beta (t_1 c + m)}{(t_1 c +k)(t_1 c +k+1)} &= t_1 \beta (t_1 c + m)\sum_{k=0}^{m-1} \left[\frac{1}{(t_1 c +k)} - \frac{1}{(t_1 c +k+1)} \right] \\
    & = t_1 \beta (t_1 c + m) \left( \frac{1}{t_1 c} - \frac{1}{t_1 c + m} \right)= \frac{m \beta}{c}.
\end{align*}
Thus,
\begin{align*}
    C^m(\beta, \delta, t) &= \frac{m \beta}{c} \left[ 1 + \frac{(t-1)c+m}{(t-1)c} + \cdots + \prod_{t_1 = \tau +1}^{t-1}\frac{t_1 c + m}{t_1 c} \right] \\
    &= \frac{m \beta}{c}\sum_{t_1 = \tau + 1}^{t}\frac{\Gamma(t+\frac{m}{c})\Gamma(t_1 )}{\Gamma(t_1 +\frac{m}{c})\Gamma(t)}.
\end{align*}
Using properties of the gamma function, we can rewrite
\begin{align}
    \nonumber
    C^m(\beta, \delta, t) &= \frac{m \beta \Gamma(t+\frac{m}{c})}{c\Gamma(t)} \left( \frac{1}{\frac{m}{c}-1} \right) \times \sum_{t_1 = \tau + 1}^{t} \left( \frac{\Gamma(t_1 )}{\Gamma(t_1-1 +\frac{m}{c})} - \frac{\Gamma(t_1 +1)}{\Gamma(t_1 + \frac{m}{c})} \right) \\
    \nonumber
    &= \frac{m \beta}{m+\beta+\delta} \frac{\Gamma(t+\frac{m}{2m+\beta+\delta})}{\Gamma(t)}\left( \frac{\Gamma(t+1)}{\Gamma(t + \frac{m}{2m+\beta+\delta})} - \frac{\Gamma(\tau +1)}{\Gamma(\tau + \frac{m}{2m+\beta+\delta})} \right)\\
    \label{eq:D-tau-C}
    &= \frac{m \beta t }{m+\beta+\delta} - \frac{m \beta \tau }{m+\beta+\delta} \frac{\Gamma(t+\frac{m}{2m+\beta+\delta})\Gamma(\tau)}{\Gamma(t)\Gamma(\tau + \frac{m}{2m+\beta+\delta})}.
\end{align}
Substituting \eqref{eq:D-tau-C} into \eqref{eq:D-tau-Gamma+C}, we get \eqref{eq:D-tau}. 
In particular,
$ \frac{m \beta }{m+\beta+\delta} \to m$ and $\frac{m}{2m+\beta+\delta} \to 0$ as $ \beta \to \infty$, so that
\begin{align*}
    \lim_{\beta \to \infty} \E[D_\tau(t) + \delta] 
    = mt +(m+\delta - m \tau),
\end{align*}
which implies that all the incoming edges are expected to connect to the anomaly, and we get \eqref{eq:D-tau-beta-infty}. 

For general $ j \in [m]$, we can extend \eqref{eq:D-tau} to find $ \E [D_\tau(t,j) + \delta ] $. Applying the recursive approach,
\begin{align*}
    \E &[D_{\tau} (t,j)+ \delta \mid D_{\tau} (t,j-1)] \\
    &= D_{\tau}(t,j-1) + \delta + \frac{(t-1) \beta+ D_{\tau}(t,j-1) +\delta}{(t-1)(2 m+ \beta + \delta)+ j-1} \\
    & = (D_{\tau}(t-1,m) + \delta) \times \frac{t-1 + \frac{j}{2m + \beta+\delta}}{t-1} + \frac{\beta j}{2m + \beta+\delta},
\end{align*}
and taking expectation, we have
\begin{align*}
    \E [D_{\tau} (t,j)+ \delta] = \E[D_{\tau}(t-1,m) + \delta] \left( 1 + \frac{j}{(t-1)(2m + \beta+\delta)} \right) + \frac{\beta j}{2m + \beta+\delta}.
\end{align*}
\end{proof}
We obtain \eqref{eq:D-tau} by solving this recursion.
\qed
\medskip
 
Figure \ref{fig: degree of anomaly} shows the growth of $D_{\tau} (t)$ over time, with our theoretical result for $ \E [D_{\tau} (t)+ \delta]$ depicted as yellow line. It is very interesting that the coefficient of the linear growth, $\frac{\beta}{m+\beta+\delta}$ is larger than the probability $\frac{\beta}{2m+\beta+\delta}$ in \eqref{rule-2} that an edge attaches itself to an anomaly. Indeed, our assumption that vertices may attach to the anomaly also through the PA mechanism, has increased the rate of growth (rather than, say, giving rise to an extra polynomial term as we conjectured at the beginning). 
\begin{figure}
    \centering
        \includegraphics[width=0.6\linewidth]{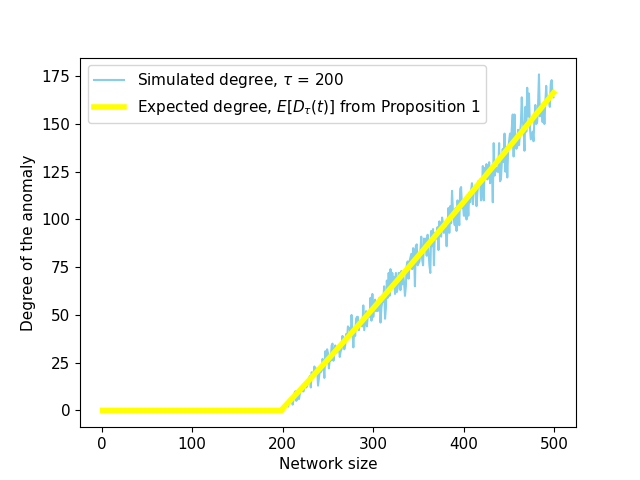}
    \caption{ The degree of the anomaly as a function of time $t$. 
    The parameters used are $\tau = 200, \delta = 0, \beta = 2.0, m = 1.$}
    \label{fig: degree of anomaly}
\end{figure}
To further explore the relation between $ D_{\tau} (t)$ and the network size, we assume that $ t = a \tau$, $ a \geq 1$. By Stirling's formula,
\begin{align}
\label{Stirling}
    \frac{\Gamma(t + a)}{\Gamma (t)} = t^{a}(1+O(1/t)),
\end{align}
when $t \to \infty$ and $a$ is fixed (see, e.g., \cite[(8.3.9)]{Hofstad_2016}). This approximation allows us to simplify the formula of expected degree of $v_\tau$. Indeed, applying it to \eqref{eq:D-tau} gives
\begin{align*}
    f(a) = \lim_{\tau \to \infty} \frac{\E [D_{\tau}(a \tau) + \delta]}{a \tau} = \frac{m\beta}{m+\beta+\delta}\left( 1- a^{-\frac{m+\beta+\delta}{2m+\beta+\delta}} \right),
\end{align*}
and
\begin{align*}
    f(1) &= 0, \qquad \lim_{a \to \infty} f(a) = \frac{m\beta}{m+\beta+\delta}, \\
    f'(a) &= \frac{m \beta}{2m+\beta+\delta} a^{-\frac{m+\beta+\delta}{2m+\beta+\delta}-1}, 
    \qquad f'(1) = \frac{m \beta}{2m+\beta+\delta}.
\end{align*}
We see that larger values of $a$, corresponding to an earlier anomaly, result in a larger deviation from the linear growth of $ \E[D_{\tau}(a \tau)]$. 

\subsection{Expected degree of the ordinary vertices}
\label{ssec:degree-ordinary-expected}
Despite the presence of an anomaly, ordinary vertices continue to receive edges based on the preferential attachment mechanism. We apply the recursive approach from  \cite[Section 8.3]{Hofstad_2016} to calculate their expected degrees, as the anomaly does not alter the edge assignment mechanism for these vertices. Regarding the attachment function, it differs for $ i \leq t \leq \tau$ and $t>\tau$. Therefore, we will analyze the following two scenarios:
\begin{enumerate}
\item[(1)] If  $ i < \tau$, for $ t < \tau$, the anomaly has not yet occurred, the expected degree of $ v_i $ is
\begin{align}
\E [D_i(G_{t}) + \delta]=
\begin{cases}
    \displaystyle
    (2m+\delta) \frac{\Gamma(t + \frac{m}{2m+\delta})}{\Gamma(1 + \frac{m}{2m+\delta}) \Gamma(t)},& i=1,\\
    \displaystyle(m+\delta) \frac{\Gamma(t + \frac{m}{2m+\delta}) \Gamma(i)}{\Gamma(i + \frac{m}{2m+\delta}) \Gamma(t)}, & 1<i<\tau.
\end{cases}
\end{align}
{See, e.g., \cite[Exercise 8.14]{Hofstad_2016} for the case where $m=1$.} Remarkably, for the model chosen here, the formula for $m>1$ is actually quite nice.

For $ t > \tau $, the attachment rule is changed after the anomaly appears. Then, the expected degree of $ v_i $ changes into 
\begin{equation}
\begin{aligned}
    \label{E(D_i)}
     \E [D_i(t)+\delta ] = 
     \begin{cases}
     \displaystyle
     (2m+\delta) \frac{ \Gamma(\tau + \frac{m}{2m+\delta})\Gamma(t  + \frac{m}{2m+\beta+\delta})}{\Gamma(1+ \frac{m}{2m+\delta})\Gamma(\tau+ \frac{m}{2m+\beta+\delta})\Gamma(t)}, & i=1,\\
     \displaystyle(m+\delta)\frac{\Gamma(\tau + \frac{m}{2m+\delta})\Gamma(t  + \frac{m}{2m+\beta+\delta})\Gamma(i)}{\Gamma(i+ \frac{m}{2m+\delta})\Gamma(\tau+ \frac{m}{2m+\beta+\delta})\Gamma(t)}, & 1<i<\tau.
     \end{cases}
\end{aligned}
\end{equation}

\item[(2)] If $ i > \tau$, for $ t > \tau$, the expected degree of $ v_i$ is 
\begin{equation}
 \E [D_i(t) + \delta] = (m+\delta) \frac{\Gamma(t+\frac{m}{2m+\beta+\delta})\Gamma(i)}{\Gamma(i+\frac{m}{2m+\beta+\delta})\Gamma(t)}.
 \end{equation}
\end{enumerate}

\subsection{Convergence of degrees for ordinary vertices}
\label{ssec:degree-ordinary-convergence}

In \cite[Section 8.3]{Hofstad_2016}, it is proved that the degree of vertices in a standard PA network scales as $ t^{\frac{1}{2 + \delta}} $ when $ m =1$, {and as $t^{\frac{1}{2 + \delta/m}} $ when $ m \geq 2$ (see \cite[Exercise 8.13]{Hofstad_2016}).}{} In our model, if $ t< \tau$, the anomaly is not included in $ G_t$, so that again the degrees of vertices are of the order $ t^{\frac{1}{2 + \delta/m}} $. When $ t> \tau $, it is necessary to analyze the convergence of the degrees of vertices separately.

For the vertices added before $\tau$, we consider the sequence $(M^{\sss\rm(1)}_i(t))_{t\geq i}$ given by
    $$  M^{\sss\rm(1)}_i(t)=\frac{D_i(t) + \delta}{\E [D_i(t) + \delta]}.$$
It is easy to see that $(M^{\sss\rm(1)}_i(t))_{t\geq i}$ is a non-negative martingale. Indeed, since $ m \leq D_i(t) < 2mt $, we have $ \E [|M^{\sss\rm(1)}_i(t)|] < \infty$. Computing the conditional expectation, we get 
    \begin{align*}
        \E  &[M^{\sss\rm(1)}_i(t+1)\mid M^{\sss\rm(1)}_i(t)] = \E [M^{\sss\rm(1)}_i(t+1)\mid D_i(t)] \\
        &= \frac{\E[D_i(t+1) + \delta \mid D_i (t)]}{\E[D_i(t+1) + \delta]} \\
        &= \frac{(D_i(t) +\delta)}{\E[D_i(t+1) + \delta]} \prod_{j=1}^{m} \left( 1 + \frac{1}{t(2m + \beta + \delta)+ j-1} \right) \\
        &= \frac{D_i(t)+\delta}{\E[D_i(t+1) + \delta]} \frac{t+\frac{m}{2m+\beta+\delta}}{t}\\
        &= \frac{D_i(t) + \delta}{\E[D_i(t) + \delta]} =M^{\sss\rm(1)}_i(t),
    \end{align*}
since also
    \eqn{
    \E[D_i(t+1) + \delta]=\E[D_i(t) + \delta]\frac{t+\frac{m}{2m+\beta+\delta}}{t}.
    }
Thus, $(M^{\sss\rm(1)}_i(t))_{t\geq i}$ is a non-negative martingale with respect to $ \left( G_t \right)_{t \geq i}$. According to the martingale convergence theorem, $M^{\sss\rm(1)}_i(t)$ converges almost surely to a limiting random variable as $ t \to \infty$ \cite[Theorem 2.24] {Hofstad_2016}, consequently, the result can be extended to establish the convergence of degrees when $ i < \tau $.  By the Stirling's formula, for sufficiently large $ t $ and $\tau $,
\begin{align*}
   \frac{D_i(t) + \delta}{\left( \frac{t}{\tau} \right)^{\frac{m}{2m+\beta+\delta}} \tau^{\frac{m}{2m+\delta}}} &= M_i^{\sss\rm(1)}(t) \frac{(d+\delta)\Gamma(i)}{\Gamma(i + \frac{m}{2m+\delta})} (1+o(1)) \\ &\underrightarrow{\scriptstyle{a.s.}}  \frac{(d+\delta)\Gamma(i)}{\Gamma(i + \frac{m}{2m+\delta})} \xi_i^{(1)},
\end{align*}
where $\xi_i^{(1)}$ is the almost sure limit of $M_i^{\sss\rm(1)}(t)$, and $d=2m$ for $i=1$, and $d=m$ for $i>1$. Thus, when $i<\tau$, $ \frac{D_i(t) + \delta}{\left( \frac{t}{\tau} \right)^{\frac{m}{2m+\beta+\delta}} \tau^{\frac{m}{2m+\delta}}}$ converges almost surely as $t \to \infty$ and $\tau \to \infty$.

Similarly, for the vertices added after $\tau$, let $(M^{\sss\rm(2)}_i(t))_{t\geq i}$ be given by
$$ M^{\sss\rm(2)}_i(t)=\frac{D_i(t) + \delta}{\E [D_i(t) + \delta]}, \qquad\qquad i > \tau.$$
Following the previously described steps again, we get $(M^{\sss\rm(2)}_i(t))_{t\geq i}$ is a non-negative martingale, and when $ t $ is large enough, 
\begin{align*}
    \frac{D_i(t) + \delta}{t^{\frac{m}{2m+\beta+\delta}} } &= M_i^{\sss\rm(2)}(t) \frac{(m+\delta)\Gamma(i)}{\Gamma(i + \frac{m}{2m+\beta+\delta})} (1+o(1))\\
    &\underrightarrow{{\scriptstyle a.s.}} \frac{(m+\delta)\Gamma(i)}{\Gamma(i + \frac{m}{2m+\beta+\delta})}\xi_i^{(2)},
\end{align*}
where $\xi_i^{(2)}$ is the almost sure limit of $M_i^{\sss\rm(2)}(t)$. Thus, when $i > \tau$, $ \frac{D_i(t) + \delta}{t^{\frac{m}{2m+\beta+\delta}}}$ converges almost surely. 

\section{Heuristic derivation of the limiting degree distribution for ordinary vertices}
In this section, we aim to investigate the limiting degree distribution of the ordinary vertices in a PA network containing an anomaly. For the standard PA network, the exponent of power law degree is $ 3+\frac{\delta}{m}$. The rigorous proof in \cite[Chapter 8]{Hofstad_2016} strongly relies on the recursive relation between the fractions of vertices with degree $k$. However, it is hard to adapt this approach, due to the alteration in the recursion after the appearance of an anomaly. Consequently, the rigorous derivation remains an open problem for future research. Here, we present a heuristic argument \cite{litvak2023randomness} that yields results consistent with our numerical simulations. 
\smallskip

The main idea is that if we can find the range of vertex index $ i $ such that the average degree of vertex $ v_i$ falls in $ (k-0.5, k+0.5) $, then the limiting fraction $p_k$ of vertices of degree $k$ can be evaluated as that range divided by $t$, as $t\to\infty$. We will explore the power-law degree distribution in different scaling regimes for the arrival time of anomaly $\tau$ as a function of time $t$. Specifically, we consider three cases: when the anomaly arrives late, mid-way or early.

\subsection{Late anomaly:  $\tau=t-t^\gamma, \gamma \in (0,1)$.}
\label{ssec:late anomaly case}
If the arrival time of $ v_\tau$ is  {quite late, for example when $ \tau = t - t^\gamma$, where $ \gamma \in (0,1)$, then only a vanishing fraction of  vertices arrives after the anomaly. Therefore, the degree distribution will be defined by the vertices that arrived before the anomaly. When $i< \tau $, the expected degree of $ v_i$ is \eqref{E(D_i)}, here we only consider $ i > 1$, as $ t $ and $ \tau $ are large enough, by the Stirling formula, $\E [D_i(t) + \delta]$ can be approximated as
\begin{equation}
    \E [D_i(t)+\delta ] = (m+\delta) \left( \frac{t}{t-t^\gamma} \right)^{\frac{m}{2m+\beta+\delta}}\left( \frac{t - t^\gamma}{i} \right)^{\frac{m}{2m+\delta}} (1+o(1)), \qquad i < \tau.
\end{equation}
Assume that $ \E [D_i(t) + \delta ] $ falls in the interval $ (k + \delta -0.5, k+ \delta+0.5)$, then $i$ should be in
$$ \left( (m+\delta)^{2+\frac{\delta}{m}} \frac{(1-t^{\gamma-1})^{\frac{\beta}{2m+\beta+\delta}} t}{(k+\delta+0.5)^{2+\frac{\delta}{m}}},(m+\delta)^{2+\frac{\delta}{m}}  \frac{(1-t^{\gamma-1})^{\frac{\beta}{2m+\beta+\delta}} t}{(k+\delta-0.5)^{2+\frac{\delta}{m}}} \right).$$
Let $ L_1 $ be the length of interval,
\begin{align*}
    L_1 &= (m+\delta)^{2+\frac{\delta}{m}} (1-t^{\gamma-1})^{\frac{\beta}{2m+\beta+\delta}} t \left(  \frac{1}{(k+\delta-0.5)^{2+\frac{\delta}{m}}} - \frac{1}{(k+\delta+0.5)^{2+\frac{\delta}{m}}} \right)\\
    &= (m+\delta)^{2+\frac{\delta}{m}} (1-t^{\gamma-1})^{\frac{\beta}{2m+\beta+\delta}} t \frac{k^{2+\frac{\delta}{m}}\left[ \left( 1 + \frac{\delta + 0.5}{k} \right)^{2+\frac{\delta}{m}}  - \left( 1 - \frac{  0.5 -\delta}{k} \right)^{2+\frac{\delta}{m}}\right]}{k^{ 4+\frac{2\delta}{m}}\left( 1 + \frac{\delta + 0.5}{k} \right)^{2+\frac{\delta}{m}}\left( 1 - \frac{  0.5 -\delta}{k} \right)^{2+\frac{\delta}{m}}}.
\end{align*}
Next we apply the binomial theorem to simplify the expression of $L_1$, to obtain, as $ k \to \infty$,
\begin{align*}
    L_1 = \left( 2+\frac{\delta}{m} \right) (m+\delta)^{2+\frac{\delta}{m}}(1-t^{\gamma-1})^{\frac{\beta}{2m+\beta+\delta}} t k^{-(3+\frac{\delta}{m})} (1+o(1)),
\end{align*}
so that, for $k\to\infty$,
    \begin{equation}
    \label{eq:pk-late}
    p_k^{\sss\rm(late)} = \lim_{t\to\infty}\frac{L_1}{t}  = \left( 2+\frac{\delta}{m} \right) (m+\delta)^{2+\frac{\delta}{m}} k^{-(3+\frac{\delta}{m})}  (1+o(1)).\end{equation}
Thus, the asymptotic degree distribution remains unchanged compared to the standard PA network. 

\subsection{Mid-way anomaly: $\tau=\alpha t$, $\alpha\in(0,1)$}
\label{ssec:mid-way anomaly case}
Assume that the arrival time of $ v_{\tau} $ scales linearly with $ t$, that is, $ \tau = \alpha t$, where $ \alpha \in (0,1)$. For $ 1<i<\tau$, the average degree is approximated by
\begin{equation}
    \E [D_i(t)+\delta ] = (m+\delta) \alpha^{- \frac{m}{2m+\beta+\delta}}\left( \frac{\alpha t}{i} \right)^{\frac{m}{2m+\delta}} (1+o(1)), \qquad i < \tau.
\end{equation}
Similarly, for the vertex $ v_i$ with $ i > \tau$, 
\begin{equation}
\label{eq:E-midway-after-tau}
    \E [D_i(t)+\delta ] = (m+\delta) \left( \frac{t}{i} \right)^{\frac{m}{2m+\beta+\delta}} (1+o(1)).
\end{equation}
We see that the average degree of the vertices born after $\tau$ grows slower than that of the vertices that arrived before $\tau$. Moreover, the right-hand side of \eqref{eq:E-midway-after-tau} is bounded, and thus it cannot fall into the interval $(k+\delta-0.5,k+\delta+0.5)$ for large $k$.   Therefore, we perform derivations for the vertices that have arrived before time $ \tau $ as the higher order probability to achieve a large degree $k$. Repeat the steps in Section~\ref{ssec:late anomaly case}, we get, as $k\rightarrow \infty,$
\begin{equation}
\label{eq:pk-midway}
p_k^{\sss\rm(mid-way)} =  \left( 2+\frac{\delta}{m} \right) (m+\delta)^{2+\frac{\delta}{m}}\alpha^{\frac{\beta}{2m+\beta+\delta}} k^{-(3+\frac{\delta}{m})} (1+o(1)).
\end{equation}
From this formula we see that the vertices arriving before $ \tau $ follow the power law distribution with exponent $ 3+\frac{\delta}{m} $. This is the same power-law exponent as in the standard PA network. However, we get a factor $\alpha^{\frac{\beta}{2m+\beta+\delta}}$ in front. This conforms with the intuition that the anomaly slows down the degree growth of high-degree vertices. Moreover, this factor decreases with $\beta$, as larger $\beta$ increases the effect of the anomaly.

\subsection{Early anomaly: $\tau =  t^\gamma, \gamma \in (0,1)$.}
\label{ssec:early anomaly case}
Suppose that the anomaly arrives quite early at $ \tau = t^{\gamma}$, where $ \gamma \in (0,1)$. Then the fraction of vertices born before $\tau$  among all vertices is vanishing. Therefore, we investigate the behavior of vertices born after $\tau$. Repeating the steps in Section~\ref{ssec:late anomaly case}, we obtain an asymptotic degree distribution given by
    \begin{equation}
    \label{eq:pk-early}
    p_k^{\sss\rm(early)} = \left( 2+ \frac{\beta}{m} + \frac{\delta}{m} \right)(m+\delta)^{2+\frac{\beta}{m} + \frac{\delta}{m}} k^{-(3+\frac{\beta}{m} + \frac{\delta}{m})}(1+o(1)), \; k \to \infty.\end{equation}
The power-law exponent of the degree distribution changes to $ 3+\frac{\beta}{m} + \frac{\delta}{m}$. We see that the anomaly has altered the power-law exponent, as also observed in the superstar model \cite{bhamidi2015superstar}. 

\subsection{Comparison to the empirical degree distribution}
Figure \ref{fig:degree-distribution-late-anomaly} - \ref{fig:degree-distribution-early-anomaly} show the empirical and the theoretically predicted  degree distributions for the PA network with the late anomaly, the mid-way anomaly, and the early anomaly, respectively.

\begin{figure}[h]
    \centering
    \begin{minipage}{0.45\textwidth}
        \centering
        \includegraphics[width=\linewidth]{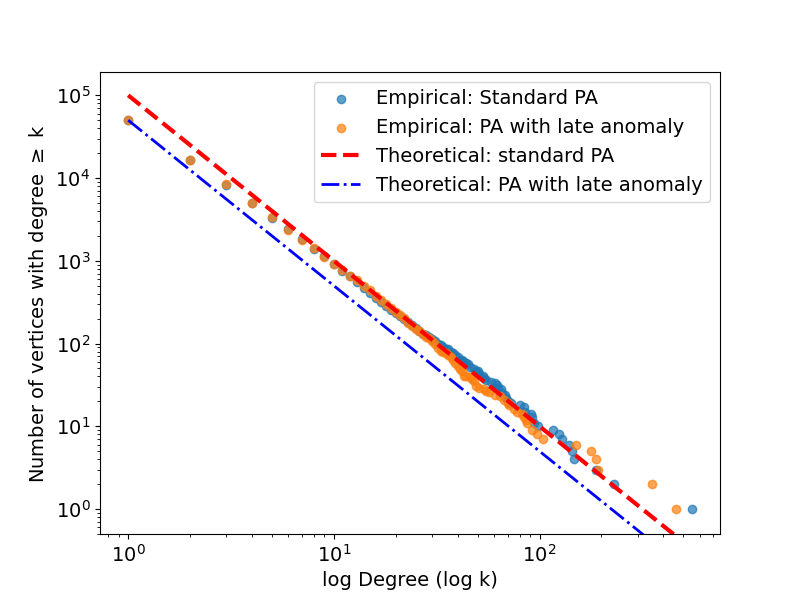}
        \subcaption{$ m =1, \beta = 5.0, \delta = 0$.}
        \label{fig:degree-distribution-late-anomaly with small m}
    \end{minipage}
    \hspace{0.5cm}
    \begin{minipage}{0.45\textwidth}
        \centering
        \includegraphics[width=\linewidth]{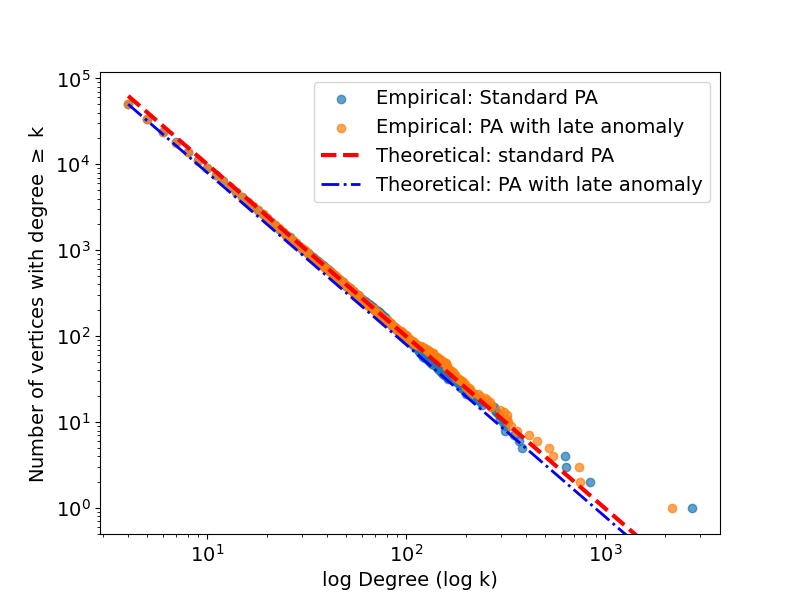}
        \subcaption{$ m =4, \beta = 10.0, \delta = 0$.}
        \label{fig:degree-distribution-late-anomaly with bigger m}
    \end{minipage}
    \caption{The complementary cumulative degree distribution for PA network with late anomaly, parameters are $t = 50000, \tau = 49950, \gamma = 0.3615$.}
    \label{fig:degree-distribution-late-anomaly}
\end{figure}

\begin{figure}[h]
    \centering
    \begin{minipage}{0.45\textwidth}
        \centering
        \includegraphics[width=\linewidth]{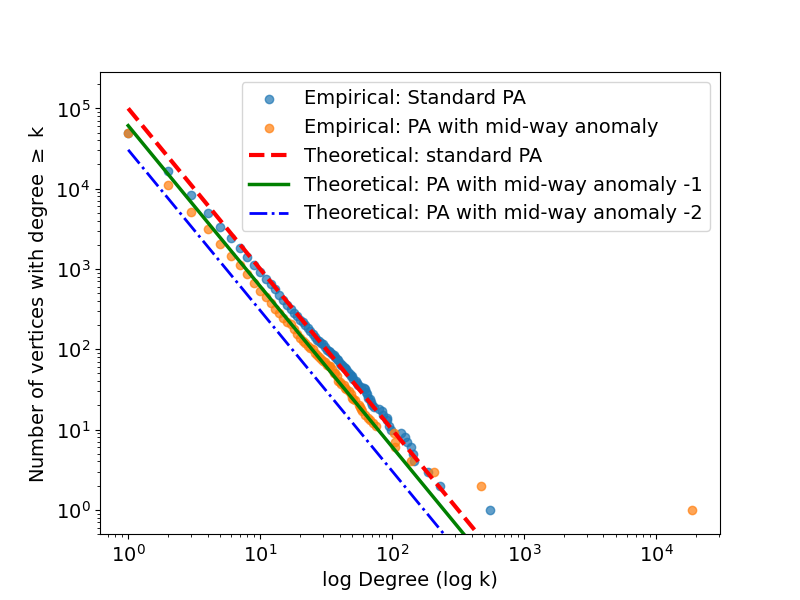}
        \subcaption{$ m =1, \beta = 5.0, \delta = 0$.}
        \label{fig:degree-distribution-midway-anomaly with small m}
    \end{minipage}
    \hspace{0.5cm}
    \begin{minipage}{0.45\textwidth}
        \centering
        \includegraphics[width=\linewidth]{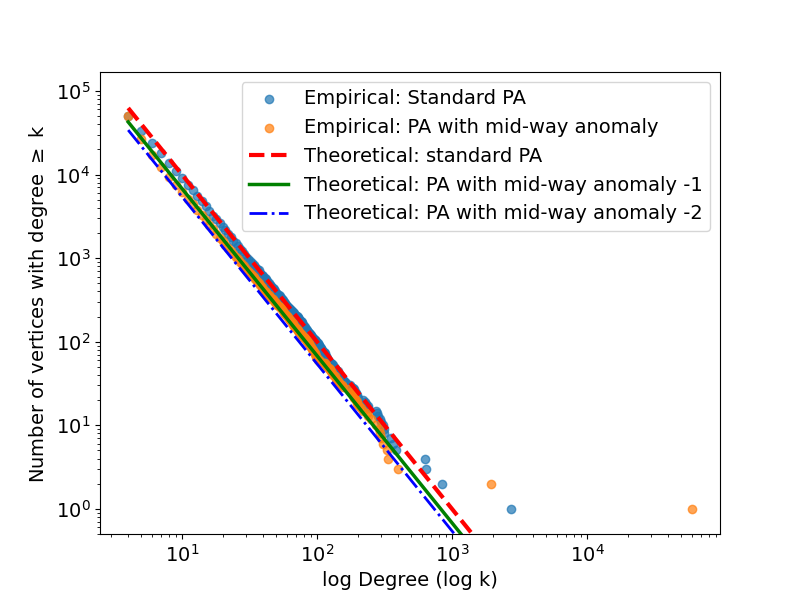}
        \subcaption{$ m =4, \beta = 10.0, \delta = 0$.}
        \label{fig:degree-distribution-midway-anomaly with bigger m}
    \end{minipage}
    \caption{The complementary cumulative degree distribution for PA network with mid-way anomaly, parameters are $t = 50000, \tau = 25000, \alpha = 0.5$. The green line uses the formula (8.4.11) in \cite[Chapter 8]{Hofstad_2016}, multiplied by factor $\alpha^{\frac{\beta}{2m+\beta+\delta}}$ as in \eqref{eq:pk-midway}. The dashed blue line follows \eqref{eq:pk-midway}.  }
    \label{fig:degree-distribution-midway-anomaly} 
\end{figure}
\vspace{-20pt}
\begin{figure}[h]
    \centering
    \begin{subfigure}{0.45\textwidth}
        \centering
        \includegraphics[width=\linewidth]{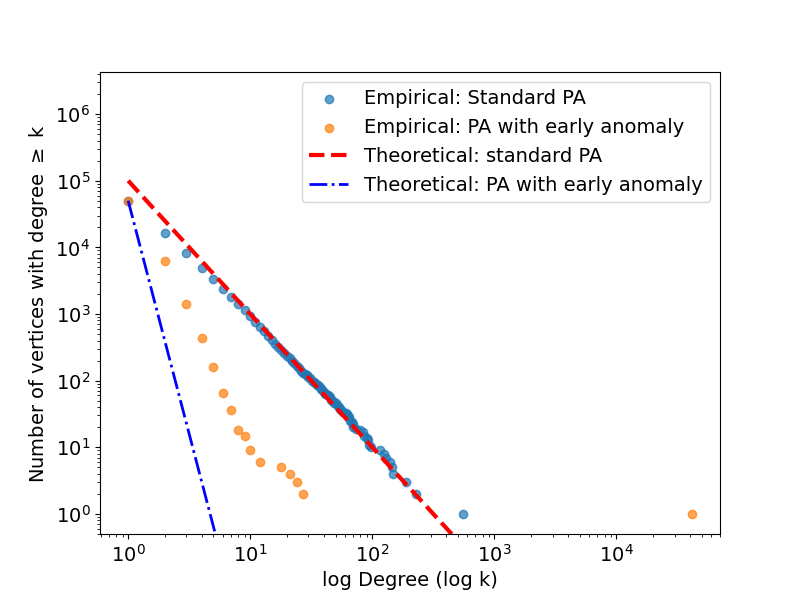}
        \subcaption{$ m =1, \beta = 5.0, \delta = 0$.}
        \label{fig:degree-distribution-early-anomaly with small m}
    \end{subfigure}
    \hspace{0.5cm}
    \begin{subfigure}{0.45\textwidth}
        \centering
        \includegraphics[width=\linewidth]{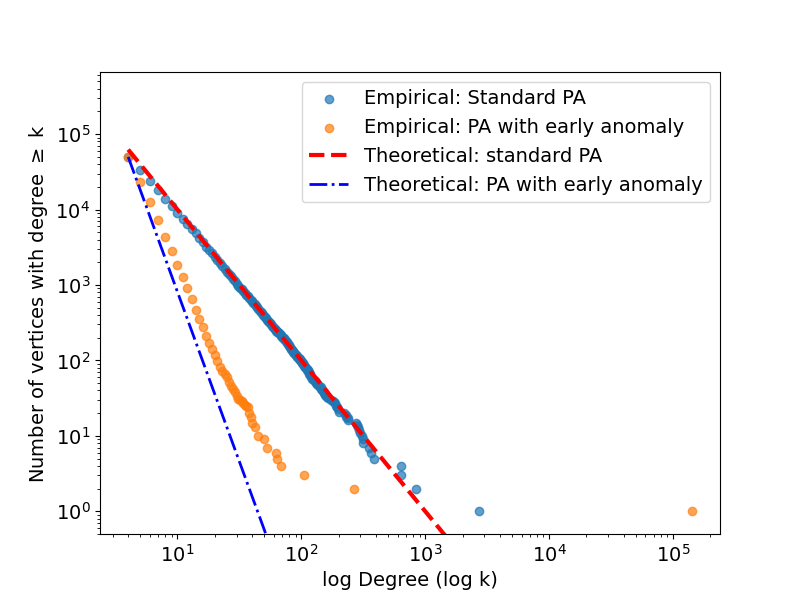}
        \subcaption{$ m =4, \beta = 10.0, \delta = 0$.}
        \label{fig:degree-distribution-early-anomaly with bigger m}
    \end{subfigure}
    \caption{The complementary cumulative degree distribution for PA network with early anomaly, parameters are $t = 50000, \tau = 50, \gamma = 0.3615$.}
    \label{fig:degree-distribution-early-anomaly}
\end{figure} 

Generally, we see that our computations correctly predict the slope, but the multiplicative factor may deviate from the experiments. In the future, a rigorous derivation for the mid-way and early anomaly is needed. The late anomaly is equivalent to standard PA network but we plot the line that we derived in \eqref{eq:pk-late} to show the difference between the correct multiplicative factor and that resulting from the heuristic derivation.

In Figure \ref{fig:degree-distribution-late-anomaly} we see that when the anomaly arrives near the end of the network's growth, the proportion of vertices with degree at least $k$ is close to that of the standard PA network. The straight line derived from \eqref{eq:pk-late} (dashed blue line) has the correct slope, but is slightly different from the empirical degree distribution. The dashed red line, based on the more precise formula (8.4.11) in \cite[Chapter 8]{Hofstad_2016} for the standard preferential attachment model has both the slope and the multiplicative factor matching the experiments. 

In Figure \ref{fig:degree-distribution-midway-anomaly} for the mid-way anomaly, the proportion of vertices with degree at least $k$ is close to that for the standard PA network, but differs by a factor smaller than~1 as predicted by \eqref{eq:pk-midway}. Interestingly, the experiments agree with formula (8.4.11) in \cite[Chapter 8]{Hofstad_2016} multiplied by the factor $\alpha^{\frac{\beta}{2m+\beta+\delta}}$. Thus, our heuristic derivation correctly predicts the effect of a mid-way anomaly. The visible outlier in Figure \ref{fig:degree-distribution-midway-anomaly} is the anomaly itself. 

In Figure~\ref{fig:degree-distribution-early-anomaly} for the early anomaly, the proportion of vertices with degree greater than $k$ is significantly lower, as the anomaly attracts a fraction of edges from the very beginning. However, \eqref{eq:pk-early} predicts an even  steeper slope. We also note that the tail of the distribution deviates strongly to the right from the straight line. We will come back to this phenomenon in the next section, where we analyze the behavior of the oldest ordinary vertex.

\section{Degree growth of the oldest vertex}

Because the existing vertices connect to a new coming vertex with a probability proportional to its degree, it is more likely that the old vertices receive more and more edges over time, and their degrees are of the order $ t^{\frac{1}{2+ \delta/m}}$ \cite[Theorem 8.2]{Hofstad_2016}. We further explore how the degree of the oldest vertex behaves after the anomaly's occurrence, by computing the exponent of power-law degree of the initial vertex $ v_1$. By Stirling's formula, as $ t $ and $\tau$ are large enough, the expected degree of the oldest vertex is
\begin{align*}
    \E[D_1(t) + \delta] = (2m+\delta) \left( \frac{t}{\tau} \right)^{\frac{1}{2+\frac{\beta}{m} + \frac{\delta}{m}}} \tau^{\frac{1}{2+ \frac{\delta}{m}}} (1+o(1)).
\end{align*}
In different scenarios, as $ t \to \infty$, the approximate expression are shown as follows:
\begin{description}
    \item[Late anomaly:] If $ \tau = t-t^{\gamma}, \gamma\in (0,1)$, 
    $$ \E [D_1(t) + \delta] = (2m+\delta)  t^{\frac{1}{2+ \frac{\delta}{m}}}(1+o(1)).$$
    \item[Mid-way anomaly:] If $\tau = \alpha t, \alpha \in (0,1)$, 
    $$ \E [D_1(t) + \delta] = (2m+\delta) \alpha^{\frac{1}{2+ \frac{\delta}{m}}-\frac{1}{2+\frac{\beta}{m} + \frac{\delta}{m}}} t^{\frac{1}{2+ \frac{\delta}{m}}}(1+o(1)).$$
    \item[Early anomaly:] If $ \tau = t^{\gamma}, \gamma\in (0,1)$,
    $$ \E [D_1(t) + \delta] = (2m+\delta) t^{\frac{\gamma}{2+ \frac{\delta}{m}}+\frac{1-\gamma}{2+\frac{\beta}{m} + \frac{\delta}{m}}}(1+o(1)).$$
\end{description}

Interestingly, the mid-way and the early anomaly affect the mean degree of the oldest vertex in a different way. 
The mid-way anomaly reduces the mean degree by a constant factor, while an early anomaly changes the exponent of $t$. Moreover, in the case of an early anomaly, $\E [D_1(t) + \delta]$ grows more quickly than $t^{\frac{1}{2+\frac{\beta}{m}+\frac{\delta}{m}}}$. The latter expression would be consistent with the power-law distribution for an early anomaly in \eqref{eq:pk-early} meaning that the oldest ordinary vertices have higher degrees than predicted by \eqref{eq:pk-early}; therefore, we see many outliers to the right in Figure \ref{fig:degree-distribution-early-anomaly}. 

\section{Conclusion and Further Research}

We introduced a PA network incorporating an anomaly, where the attachment rule of the anomaly is vastly different from that of the ordinary vertices. We derived insights on the growth of the degrees and their distribution in this model. Below we list some directions for further research:

\begin{enumerate}
    \item[(1)] {\bf Concentration of the degree of an anomalous vertex.} We see in Figure~\ref{fig: degree of anomaly} that the degree of the anomaly is close to its mean. The martingale convergence theorem is a standard way to prove such concentration. However, we could not directly apply this method to the degree of the anomaly due to the presence of a linear term in its expression. 
    \item[(2)] {\bf Convergence of degree sequences.} In \cite[Chapter 8]{Hofstad_2016}, the limiting degree distribution is derived from the recursive equations for the number of vertices of degree $k$. However, we could not use this method because the occurrence of an anomaly changes the recursion.  In previous works \cite{banerjee2023fluctuation,bhamidi2018longrange}, the asymptotic degree distribution of PA networks with change points was rigorously derived using {continuous-branching processes}, which embed the growth of the PA network for $m=1$ in continuous time. Investigating whether these techniques can be adapted to our model is a promising direction for gaining deeper insights into the interplay between anomalies and degree distributions.
    \item[(3)] {\bf Different attachment mechanism.} In this paper, we consider only the case where the anomaly attracts new edges at a {\em constant} additional rate. An interesting direction for future research would be to explore scenarios where the anomaly attaches new edges at rates that change over time, such as an increasing rate or for a randomly determined duration.
    \item[(4)] {\bf Anomaly detection.} Building on our results, detecting anomalies is a natural next step. Existing work \cite{LyapunovAnomalyDetection} has utilized Lyapunov-based method to detect certain anomalous events in PA network. Our model, where the anomaly arrives at a specific point and alters the attachment mechanisms, requires a different approach. While many studies on anomaly detection in dynamic networks focus on analyzing spatial and temporal  \cite{LiuDectectionTransformer,ZhengAddGraph2019}, or structural \cite{kipf2017GCN}, features, we find it an interesting problem to detect the anomaly based only on the history of $G_t$. Our initial attempts show that this problem is more challenging than one could expect given how strong our anomaly is. We hope to report on the progress in the near future. 
\end{enumerate}

\section*{Acknowledgments} The work of QL was supported by the China Scholarship Council. The work of RvdH and NL is supported by the
Netherlands Organisation for Scientific Research (NWO) through the Gravitation NETWORKS grant 024.002.003, and by the National Science Foudation under Grant No. DMS-1928930 while the authors were in residence at the Simons Laufer Mathematical Sciences Institute in Berkeley, California, during the Spring semester 2025.


\end{document}